\documentclass[a4paper]{article}
\usepackage[utf8]{inputenc}
\usepackage{amsmath}
\usepackage{wrapfig}
\usepackage{physics}
\usepackage{centernot}
\usepackage{geometry}
\usepackage{upgreek}

\usepackage[final]{hyperref} 

\hypersetup{colorlinks=true,
            allcolors=blue,
            citecolor=blue!75!black
            }

\usepackage{xcolor}
\definecolor{modification}{rgb}{0, 0, 0}
\NewDocumentCommand\modification{m}{\textcolor{modification}{#1}}

\usepackage{amsmath, amsthm, mathtools} 
\usepackage{amssymb}
\usepackage{csquotes}

\usepackage[eprint,          
            doi=false,
            url=false,
            backend = biber,
            sorting = none,
            maxnames=10,     
            giveninits=true, 
            safeinputenc,     
           style = nature,
           isbn = false,
           date=year]{biblatex}
\newtheorem{theorem}{Theorem}
\newtheorem{definition}{Definition}

\newtheorem{lemma}{Lemma}

\usepackage{titling}

\title{\vspace{-1cm}Gleason's Theorem for a Qubit as Part of a Composite System}
\author{Vincenzo Fiorentino and Stefan Weigert\\
\small vincenzo.fiorentino@york.ac.uk, stefan.weigert@york.ac.uk \\}

\date{%
    \normalsize  Department of Mathematics, University of York, York YO10 5GH, United Kingdom\\[2ex]%
    June 2026\protect\footnote{Originally uploaded 19 November 2025; this is the version published as \href{https://doi.org/10.1103/8dyw-6ktq}{Phys. Rev. A 113 062221 (2026)}.} 
}

\begin{document}

\maketitle

\begin{abstract}
We extend Gleason's theorem to the two-dimensional Hilbert space of a qubit by invoking the standard axiom that describes \textit{composite} quantum systems. 
The tensor-product structure allows us to derive density matrices and Born's rule for $d=2$ from a simple requirement: the probabilities assigned to measurement outcomes must not depend on whether a system is considered on its own or as a subsystem of a larger one. 
In line with Gleason's original theorem, our approach assigns probabilities only to \textit{projection-valued measures}, while other known extensions rely on considering more general classes of measurements. \modification{This extension of Gleason's theorem to two-dimensional systems is shown to remain valid for some foil theories of quantum theory.} 
\end{abstract}
\section{Introduction}

Gleason's theorem \cite{gleason1957measures} derives both the density matrix formalism for quantum states and Born's rule for calculating outcome probabilities.\textcolor{modification}{\footnote{\textcolor{modification}{The ``derivation of Born's rule" here refers to the trace rule, a mathematical prescription that generates outcome probabilities or expectations from quantum states. From a more physical perspective, Gleason's theorem establishes the Born rule as the link between the quantum formalism and empirical predictions.}}}
 The result is based on assigning non-contextual probabilities to measurement outcomes that are associated with \textit{projection-valued measures} (PVMs). 

Importantly, Gleason's argument does not hold for a single qubit, i.e.\ the quantum system associated with a Hilbert space of dimension two. From a conceptual point of view, this situation is not satisfactory since the qubit represents the building block of quantum information processing and the currently emerging quantum technologies \cite{nielsen2010quantum, gill2025quantum}.

The proof of Gleason's theorem is based on consistency conditions that arise only if projection operators appear in multiple resolutions of the identity. The structural simplicity of projective  measurements in the \textit{two-dimensional} Hilbert space of a qubit does not allow for the required ``intertwining" of projections. Therefore, Gleason's theorem cannot be used to rule out probability assignments that, for example, attribute definite outcomes to measurements of two or more non-commuting observables. Such assignments are incompatible with quantum theory. In fact, no qubit state $\ket \psi $ with $\langle \psi| \sigma_x| \psi \rangle = \langle \psi| \sigma_z | \psi \rangle = 1$ exists, where $\sigma_x$ and $\sigma_z$ are the Pauli observables proportional to the $x$- and the $z$-component of a spin $1/2$, respectively.

We will extend Gleason's theorem to a two-dimensional Hilbert space by assuming that probability assignments to measurement outcomes do not depend on whether a system such as a qubit is considered on its own or as a subsystem of a larger one. In other words, probability assignments must be defined consistently for \textit{composite} quantum systems. The main advantage of this approach is that we only invoke another standard postulate of quantum theory instead of enlarging the set of measurements to which probabilities are assigned. 

\modification{So-called \textit{Gleason-type} theorems that apply to a two-dimensional Hilbert space are based on probability assignments to the outcomes associated with sets of measurements \textit{larger} than projective ones. Examples of such sets are the elements of \textit{positive operator-valued measurements} (POVMs) \cite{busch2003quantum, caves2004gleason} or \textit{projective-simulable} measurements \cite{wright2019gleason}. In contrast, our derivation only uses \textit{projective} measurements, just as the original theorem from 1957. The main idea of both our approach---to exploit constraints on local measurements when carried out on composite systems featuring a qubit---and the result are natural from a physicist's point of view (see Sec.~3.5 of \cite{rau2021quantum}).}

Consistently assigning outcome probabilities is a physically motivated assumption, in contrast to the rather abstract definition of quantum states as density operators that act on a specific Hilbert space. In this sense, Gleason's theorem can be understood as an early contribution to the \textit{reconstruction programme} of quantum theory that aims to derive quantum theory from intuitive postulates \cite{hardy_quantum_2001, grinbaum2007reconstruction, chiribella_informational_2011,masanes_derivation_2011, jaeger2019information}. Clarifying the logical relationships among the standard postulates and possibly reducing their number also helps to identify the core of quantum theory \cite{fiorentino_beyond_2025, wilson_origin_2023, erba_composition_2024, auffeves2020deriving, auffeves2022revisiting}. Similarly, the proposed extension of Gleason's theorem represents a welcome conceptual simplification since we show that nothing but the properties of composite quantum systems feed into it. 

In Sec.\ \ref{Gleason thm}, we briefly recall Gleason's theorem and introduce suitable notation. The argument leading to Thm.\ \ref{thm: our GTT}, our main result, is presented in Sec.\ \ref{Sec: qubitGleason}, building on the description of measurements in com\-po\-site quantum systems. \modification{In Sec.\ \ref{implications}, we examine whether our derivation is ``structurally stable'': does  it also apply to \textit{foil theories} of quantum theory characterized by real or quaternionic Hilbert spaces, by alternative tensor products or state-update rules?} We conclude with a summary \modification{(Sec.~\ref{summary})} and place our results in a wider context (Sec.~\ref{sum&disc}).

\section{Gleason's theorem} \label{Gleason thm}
Gleason's theorem relies on three assumptions. The first two of them set the scene by (H) associating a finite-dimensional Hilbert space with each quantum system, and by (M) defining what measurements are.

\begin{itemize}
    \item[(H)] To each quantum system, there corresponds a complex \textit{Hilbert space} $\mathcal{H}_d$ with finite dimension $d \geq 2$.
    \item[(M)] A \textit{measurement} $M$ on a quantum system is described by a projection-valued measure (PVM), i.e.\ a set $M=\{ P_x \}_x$ of projectors $P_x\in \mathcal{P}(\mathcal{H}_d) \equiv \{ P \in \mathbb{C}^{d\times d} \, | \, P^2=P, \, P^{\dagger}=P \}$, that are mutually orthogonal ($P_xP_y=P_x \delta_{xy}$) and sum to the identity operator, $\sum_x P_x=I_d$. Each projector $P_x$ corresponds to a possible outcome $x$ of a measurement $M$, and the set of all measurements on $\mathcal{H}_d$ will be denoted by $\mathsf{M}_d$.
\end{itemize}

The third assumption (S) introduces quantum states in an operational way: states are considered to be in one-to-one correspondence with non-contextual probability assignments to the outcomes of all measurements that can be performed on a given system. The concept of a \emph{frame function respecting the measurement set} $\mathsf{M}_d$ is crucial here: it is a map $f:\mathcal{P}(\mathcal{H}_d) \to [0,1]$ that assigns probabilities to the projectors of any measurement $M \in \mathsf{M}_d$ in such a way that
\begin{equation}
    \sum_{P_x \in M} f(P_x) = 1 \, 
\end{equation}
holds. Operationally, the number $f(P_x)$ represents the probability that the measurement $M$ outputs $P_x \in M$. The outcome probabilities are \textit{non-contextual} in the sense that for any two different measurements $M$ and $M^{\prime}$ containing the outcome $P_x$, the probability $f(P_x)$ is assigned the same value.
In this terminology, Gleason's third assumption reads as follows. 
\begin{itemize}
    \item[(S)] The \textit{states} of a $d$-dimensional quantum system correspond to the frame functions $f$ that respect the measurement set $\mathsf{M}_d$.
\end{itemize}
The set of frame functions respecting the measurement set $\mathsf{M}_d$ will be denoted by $\mathcal{F}_d$.
It corresponds to all \textit{distinguishable} preparations of a $d$-dimensional quantum system, i.e.\ those that assign different outcome statistics to at least one measurement $M \in \mathsf{M}_d$. 

Gleason's theorem establishes a one-to-one correspondence between $\mathcal{F}_d$ and the set of density matrices $\mathcal{S}(\mathcal{H}_d)$ for dimension $d \geq 3$.
\begin{theorem}[Gleason's theorem] \label{thm: GT}
    Assume (H), (M) and (S). Then, for $d \geq 3$, any frame function $f \in \mathcal{F}_d$ admits an expression
    \begin{equation} \label{eq: Born rule frame function}
        f(P_x) = \mathrm{Tr}(P_x \, \rho)  \quad \text{for all } \, P_x \in \mathcal{P}(\mathcal{H}_d) \, ,
    \end{equation}
    where $\rho \in \mathcal{S}(\mathcal{H}_d)$ is a \textrm{density matrix}, i.e.\ a $d$-dimensional, non-negative (hence Hermitian) operator with unit trace.
\end{theorem}

In the two-dimensional Hilbert space $\mathcal{H}_2$, any non-negative operator with unit trace gives rise to a frame function $f$ on the set of qubit measurements $\mathsf{M}_2$, just as in any larger dimension $d\geq 3$. There are, however, probability assignments that do \textit{not} arise through a relation of the form \eqref{eq: Born rule frame function}. When $d = 2$, each projector $P_x \in \mathcal{P}(\mathcal{H}_2)$ appears only in a single measurement, $M=\{ P_x, I_2-P_x \}$. This lack of ``intertwining'' means that the values $f(P_x)$ and $f(I_2-P_x)$ can be chosen independently of all other measurements. Consequently, \textit{any} probability assignment to each pair of projectors forming the measurements $M\in \mathsf{M}_2$ defines a valid frame function for a qubit.

The properties of measurements in \textit{composite} systems will allow us, in a natural way, to rule out probability assignments that do not arise from density matrices. It is instructive to informally describe the indirect relation between Gleason's theorem and the state space of a single qubit that we will exploit. We will consider an even composite dimension to which the theorem applies, $d=4$, say. In this case, we work with the Hilbert space $\mathcal{H}_4$ used to describe a pair of qubits. Here is a family of measurements $M_{\psi}$ on $\mathcal{H}_4$ consisting of three projection operators each,
\begin{equation} \label{threeprojectors}
    M_\psi=\{ |0\rangle \langle0|\otimes I_2,\,
    |1\rangle \langle 1|\otimes |\psi \rangle \langle \psi|, \,
    |1\rangle \langle 1|\otimes |\psi^\perp \rangle \langle \psi^\perp| \}
\end{equation}
where $\ket{\psi}$ and $\ket{\psi^\perp}$ are any two orthogonal two-component vectors. The projection operator $\Pi = |0\rangle \langle0|\otimes I_2$ is, of course, closely related to the projector $\uppi = |0\rangle \langle0| \in \mathcal{P}(\mathcal{H}_2)$ associated with a measurement outcome of a \textit{single} qubit. The fact that each of the measurements $M_\psi$ contains the operator $\Pi$ means that they are intertwined. Thus, there exist constraints on assigning a probability to the outcome associated with $\Pi$, in contrast to the single-qubit projector $\uppi$ which occurs only in one measurement, namely $M=\{ |0\rangle \langle0| , |1\rangle \langle 1| \}$. It remains to show that the properties of measurements in composite systems induce constraints that limit the allowed outcome probabilities for measurements carried out on a subsystem.

\section{Gleason's theorem for a qubit} \label{Sec: qubitGleason}

To derive our main result, Thm.\ \ref{thm: our GTT}, we begin by recalling the description of composite quantum systems. Next, we justify the requirement that any measurement on a subsystem must have a unique representation within the set of measurements on the composite system. The freedom of an observer to regard any single system as a component of a larger, composite system (or not) implies a consistency condition on valid probability assignments to measurement outcomes. This constraint entails an operational definition of quantum states that is more restrictive than the one used in Gleason's theorem. Nevertheless, it is sufficient to extend the theorem to qubits and does not require assumptions for probability assignments to measurements other than projective ones.

\subsection{Composite quantum systems} \label{compositesystems}

Independent physical systems may be combined to form a single, \textit{composite} system. A theory describing the dynamics of individual  systems or their behaviour when performing measurements must also be capable of describing the properties of composite systems. Given measurement sets $\mathsf{M}^A$ and $\mathsf{M}^B$ for individual systems $A$ and $B$, the theory must define a composite measurement set $\mathsf{M}^{AB}$ for the joint system $AB$. In our context, this construction must  include a set of frame functions $\mathcal{F}^{AB}$ associated with the possible preparations of the composite system.

In quantum theory, the composition of systems employs the  tensor product. It will be sufficient to consider bipartite systems, i.e.\ composite systems built from only two parts.
\begin{itemize}
    \item[(C)] To each bipartite quantum system $AB$, there corresponds a complex Hilbert space that is obtained by tensoring the Hilbert spaces $\mathcal{H}^A$ and $\mathcal{H}^B$ of its constituents $A$ and $B$, respectively, $\mathcal{H}^{AB}=\mathcal{H}^{A}\otimes \mathcal{H}^{B}$.
\end{itemize}

Consequently, the Hilbert space representing a single system with dimension $dd^{\prime}$ is isomorphic to the Hilbert space of composite systems consisting of two parts of dimensions $d$ and $d^{\prime}$, respectively, $\mathcal{H}_{d {d^{\prime}}} \cong \mathcal{H}_{d } \otimes \mathcal{H}_{ d^{\prime}}$. In addition, there is no difference between the sets of measurements defined on these isomorphic spaces. The product structure also does not affect the definition of frame functions: Gleason's theorem is valid for \textit{any} Hilbert space of dimension larger than two, including those that describe composite systems.

\subsection{Subsystem measurements}

From an operational point of view, the defining difference between a composite and a non-composite (quantum) system is the possibility of carrying out operations on its constituents.\footnote{As put by Barrett in \cite{barrett2007information}: ``For each of the [\dots] measurements [\dots] of system $B$, there must be at least one operation on the joint system $AB$ that corresponds to performing that measurement.''} This property means that the set $\mathsf{M}^{AB}$ must contain elements that may be interpreted as representing measurements performed only on a subsystem. In other words, one ``needs to know how to represent the [measurements on] the small system as [measurements on] the larger system'' \cite{masanes2019measurement}. The basis of our argument will be that there must be a consistent description of any measurement $M^A \in \mathsf{M}^A$ carried out by an observer on a system $A$ which may, or may not, be part of a larger composite system $AB$. 

To satisfy this requirement, we need to represent the operation of \textit{measuring $M^A$ on subsystem $A$, and nothing on subsystem $B$ of the composite system AB}.\footnote{The requirement that all measurement on a single system are uniquely represented within the set of measurements on a composite system is a universal property of generalised probabilistic theories \cite{plavala2023general}.} Mathematically, we have to (injectively) embed the measurements of the set $\mathsf{M}^A$ in the set $\mathsf{M}^{AB}$ using a map  $\Phi :\mathsf{M}^A \to \mathsf{M}^{AB}$, say. Clearly, when representing the measurement $M^A \in \mathsf{M}^A$ as a measurement $\Phi(M^A)$ on the composite system $AB$, the number of possible outcomes must not change. It will be convenient to reuse the symbol $\Phi$ when extending the map from measurements to measurement outcomes. The set of projections on system $A$ should be mapped to projections on the larger system $AB$: $\Phi:\mathcal{P}(\mathcal{H}^A) \to \mathcal{P}(\mathcal{H}^{AB})$.
Accordingly, $\Phi(P_x^A)$ will describe the operation of \textit{measuring subsystem $A$ and obtaining the outcome represented by $P_x^A \in \mathcal{P}(\mathcal{H}^A)$, while measuring nothing on subsystem $B$}.

In quantum theory, the map $\Phi$ can be made explicit by invoking \textit{composition compatibility} (CC), which formalises the standard way to describe local measurements on multi-partite systems. 
\begin{itemize}
    \item[(CC)] The measurement $M^A=\{P_x^A\} \in \mathsf{M}^A$ on system $A$ being part of a composite system $AB$, is represented in the set $\mathsf{M}^{AB}$ by the measurement $M^{AB}=\{ P_x^A \otimes I^B  \}$, where $I^B$ is the identity operator on $\mathcal{H}^B$.
\end{itemize}
It is reasonable to assume a map of the form \begin{equation} \label{CC in math form}
    \Phi(P_x^A) = P_x^A \otimes I_d^B  \, ,
\end{equation}
because $\{I_d^B\}$ represents the ``trivial measurement'' on the Hilbert space  $\mathcal{H}_d^B$ of a $d$-dimensional system $B$. It reveals no information about the preparation of the system $B$ since $f(I^B_d)=1$ for all $f\in \mathcal{F}_{d}^B$. 
Thus, Eq.\ \eqref{CC in math form}, i.e.\ property (CC), formalises the assumption that measurements on $A$ alone \textit{cannot} yield information about subsystem $B$.\footnote{The form of the map $\Phi$ in Eq.\ \eqref{CC in math form} generalises beyond quantum theory. In any generalised probabilistic theory, given an effect $e^A$ of system $A$ and the unit effect $u^B$ of system $B$ (where $u^B(s)=1$ for all normalised states $s$ of $B$), the product $e^A \otimes u^B$ is the effect of the composite system $AB$ that represents a local measurement of $e^A$ on $A$ \cite{plavala2023general}. Such product effects are included in the theory regardless of the particular tensor product chosen to describe system composition.}

Having identified the elements of the set $\mathsf{M}^{AB}$ that correspond to measurements on a subsystem, we exploit the fact that each preparation of a composite system $AB$ determines a unique preparation of its subsystems. Let us denote frame functions of the composite system by $F\in \mathcal{F}^{AB}$. Then, each $F$ will induce a frame function $\phi_F \in \mathcal{F}^A$ of the subsystem $A$ when restricted to the subset $\Phi (\mathcal{P}(\mathcal{H}^A))\subset \mathcal{P}(\mathcal{H}^{AB})$, or explicitly\footnote{The relation \eqref{eq: gamma fcn (arbitrary g)} is standard when considering frame functions for composite quantum systems; see, e.g., \cite{dvurecenskij1993gleason}.},
\begin{equation} \label{eq: gamma fcn (arbitrary g)}
    \phi_F(P^A_x)=F(\Phi(P^A_x)) = F(P_x^A \otimes I^B) \qquad\text{for all } P_x^A \in \mathcal{P}(\mathcal{H}^A)\, .
\end{equation}
Clearly, $\phi_F$ is a frame function for system $A$ since 
\begin{equation} \label{eq: induced ff is a ff}
    \sum_{P_x^A \in M^A} \phi_F(P_x^A) =
    \sum_{P_x^A \in M^A} F(P_x^A \otimes I^B) =1 \quad \text{for all }M^A \in \mathsf{M}^A \, .
\end{equation}

Therefore, the map $\phi: F \mapsto \phi_F$ relates each preparation of $AB$ to a unique \textit{induced preparation} of $A$. If such a function $\phi$ did not exist, it would not make sense to consider $AB$ as a ``composite'' system containing $A$. Generally, the mapping $\phi$ will be surjective (every preparation of $A$ not considered as a subsystem of $AB$ arises from some preparation of $AB$ via $\phi$) but not injective (distinct preparations of $AB$ can yield identical preparations of $A$).

\subsection{Marginal frame functions}\label{MarginalFramFunctions}

Measurements in composite systems are thus seen to necessarily satisfy a consistency requirement. All \textit{ valid states} of system $A$ must correspond to frame functions $f \in \mathcal{F}^A$ that follow from preparations of composite systems containing $A$ as a subsystem:  after combining $A$  with any other system $B$, it must be possible to generate any subsystem frame function $f$ from a frame function $F \in \mathcal{F}^{AB}$ through a map $f= \phi_F $. This property encodes the operational principle that \textit{any single system can be regarded, without affecting the predictions of the theory, as part of a larger composite system}. Frame functions satisfying this property will be called \textit{marginal}. The following definition uses Postulate (C) and incorporates assumption (CC).\footnote{Conditions analogous to Eq.\ \eqref{cons eq with CC} in Def.\ \ref{def: CC frame function} have appeared as consistency requirements for measurement statistics in operational theories with composite systems; see, e.g., \cite{barrett2007information, masanes2019measurement, plavala2023general}.}

\begin{definition} \label{def: CC frame function}
A frame function $f \in \mathcal{F}_d^A$ is called \textit{marginal} if, for all 
 systems $B$ with Hilbert space $\mathcal{H}_{d^\prime}^B$ of finite dimension $d^{\prime} \geq 2$, there exists a frame function $F \in \mathcal{F}_{dd^{\prime}}^{AB}$ for the composite system $\mathcal{H}_{d}^A \otimes \mathcal{H}_{d^\prime}^B$ such that $f=\phi_F$, or explicitly,
\begin{equation} \label{cons eq with CC}
    f(P_x^A) = F (P_x^A \otimes I^B_{d^\prime}) \quad\text{for all } P_x^A \in \mathcal{P}(\mathcal{H}_d^A) \, .
\end{equation}
\end{definition}
The set of marginal frame functions respecting the measurements $\mathsf{M}_d$  associated with system $A$ will be denoted by $\widetilde{\mathcal{F}}_{d}^A$. 

Marginality of frame functions provides an operational definition of quantum states that is stronger than the assumption ($\mathrm{S}$) used in Gleason's theorem (cf.\ Sec.\ \ref{Gleason thm}).
\begin{itemize}
    \item[($\widetilde{\mathrm{S}}$)] The \textit{states} of a $d$-dimensional system correspond to \textit{marginal} frame functions $f\in \widetilde{\mathcal{F}}_d$ that respect the measurement set $\mathsf{M}_d$.
\end{itemize}

Using assumption ($\widetilde{\mathrm{S}}$), we will identify the \textit{partial trace} as the unique mapping $\phi$ that generates frame functions for subsystems. It then follows that we can extend Gleason's theorem to two-dimensional quantum systems without modifying postulate ($\mathrm{M}$).

Let us show now that,  for quantum systems with dimensions $d \geq 3$,  marginal frame functions in ($\widetilde{\mathrm{S}}$) give rise to the same set of quantum states as the original frame functions in ($\mathrm{S}$). 
\begin{lemma} \label{lem: d>2}
    All frame functions for systems of dimension $d \geq 3$ are marginal, i.e.\ $\widetilde{\mathcal{F}}_d = \mathcal{F}_d$.
\end{lemma}
\begin{proof}
    Since $d>2$, Gleason's Theorem implies that all valid outcome probability distributions of an isolated $d$-dimensional system can be represented by means of the trace over a suitable density matrix $\rho_f \in \mathcal{S}(\mathcal{H}_d^A)$, i.e.\ $f(P_x^A) = \text{Tr}(P_x^A \rho_f)$. Similarly, since $dd^{\prime}>2$ for any $d^{\prime}\geq 2$, we have that for each $F \in \mathcal{F}_{dd^{\prime}}^{AB}$ there exists a density matrix $\rho_F \in \mathcal{S}(\mathcal{H}_{dd^{\prime}}^{AB})$ such that
    $F(P_y^{AB}) = \text{Tr}(P_y^{AB} \rho_F)$ for all $P_y^{AB} \in \mathcal{P}(\mathcal{H}_{dd^{\prime}}^{AB})$.
Therefore, the marginality condition \eqref{cons eq with CC} reads 
\begin{equation} \label{equal traces}
        \text{Tr}(P_x^A \rho_f) = \text{Tr}[(P_x^A \otimes I^B_{d^{\prime}} ) \, \rho_F] \,,  
    \end{equation}
and it is satisfied for all $P_x^A \in \mathcal{P}(\mathcal{H}^A_d)$ if and only if $\text{Tr}_B (\rho_F) = \rho_f$. Thus, the mapping $\phi$ is uniquely given by the (surjective but not injective) partial trace over subsystem $B$, $\phi=\text{Tr}_B: \mathcal{S}(\mathcal{H}^{AB}) \to \mathcal{S}(\mathcal{H}^A)$. For any space $ \mathcal{H}_{d^{\prime}}^B$ and any $\rho_f \in \mathcal{S}(\mathcal{H}_d^A)$, there exists some $\rho_F \in \mathcal{S}(\mathcal{H}_{dd^{\prime}}^{AB})$ satisfying $\text{Tr}[(P_x^A \otimes I^B) \,\rho_F] = \text{Tr}(P_x^A \rho_f)$. It follows that for $d\geq 3$, all frame functions are marginal, i.e.\ $\widetilde{\mathcal{F}}_d= \mathcal{F}_d$.
\end{proof}

As already illustrated, there exist probability assignments to projections in the space $\mathcal{H}_2$ that do not correspond to valid density matrices for a qubit. We will now establish a one-to-one correspondence between the subset of \textit{marginal} frame functions and the set of two-dimensional density matrices $\mathcal{S}(\mathcal{H}_2)$.\footnote{In other words, the set of marginal qubit frame functions is a proper subset of all qubit frame functions, $\protect\widetilde{\mathcal{F}}_2 \subset \mathcal{F}_2$.} This is the content of  Theorem \ref{thm: our GTT}, which provides an extension of Gleason's theorem (Thm.\ \ref{thm: GT}), building on the set of marginal frame functions $\widetilde{\mathcal{F}}_d$ for composite systems, defined in postulate ($\widetilde{\mathrm{S}}$). It represents the main result of the paper.

\begin{theorem} \label{thm: our GTT}
    Assume (H), (M), (C) and ($\widetilde{S}$). Then, for $d \geq 2$, any marginal frame function $f \in \widetilde{\mathcal{F}}_d$ admits an expression 
\begin{equation}
        f(P_x) = \mathrm{Tr}(P_x \, \rho) \quad \text{for all } P_x \in \mathcal{P}(\mathcal{H}_d)\, ,
    \end{equation}
where $\rho \in \mathcal{S}(\mathcal{H}_d)$ is a $d$-dimensional density matrix. In particular, $\widetilde{\mathcal{F}}_d \cong \mathcal{S}(\mathcal{H}_d)$.
\end{theorem}

\begin{proof}
    Lemma \ref{lem: d>2} implies the claim for dimensions $d \geq 3$.

If $d=2$, we consider a composite system with Hilbert space $\mathcal{H}^{AB}_{2d^\prime}$, for any fixed $d^{\prime}\geq 2$. Theorem \ref{thm: GT} applies to this system since $2d^{\prime}>2$: hence, we can express each frame function $F \in \mathcal{F}^{AB}_{2d^\prime}$ of the composite system as $F(P_x^{AB}) = \text{Tr}(P_x^{AB} \rho_F)$, for some suitable density matrix $\rho_F \in \mathcal{S}(\mathcal{H}^{AB}_{2d^\prime})$. The marginality condition \eqref{cons eq with CC} of the qubit frame function $f\in \widetilde{\mathcal{F}}_2$ becomes 
\begin{equation} \label{partialtrace}
    f(P_x^A) = \text{Tr}[(P_x^A \otimes I^B) \, \rho_F] = \text{Tr}[P_x^A \, \text{Tr}_B(\rho_F)]\, , \quad  \text{ for any } P_x^A \in \mathcal{P}(\mathcal{H}^A_2)\,, 
\end{equation} 
providing a representation in terms of the $2\times 2$ matrix $\rho_f \equiv \text{Tr}_B (\rho_F)$. Tracing out system $B$  from a density matrix of a bi-partite systems with Hilbert space $\mathcal{H}^{AB}_{2d^\prime}=\mathcal{H}^A_2\otimes \mathcal{H}^B_{d^\prime}$ results in a non-negative matrix with unit trace of order $2$. \textit{All} such matrices can arise in this way, hence $\mathcal{S}(\mathcal{H}_2) \cong \widetilde{\mathcal{F}}_2  \subset \mathcal{F}_2$.
\end{proof}

\section{\modification{Generalizations}} \label{implications}
\textcolor{modification}{Gleason's theorem remains valid if \textit{real} instead of complex Hilbert spaces of dimension $d \geq 3$ are considered \cite{gleason1957measures}. We have shown that the theorem extends to a \textit{two-dimensional} complex Hilbert space when invoking the standard description of composite quantum systems. Does this extension of Gleason's theorem also hold for Hilbert spaces defined over other number fields?} 

\textcolor{modification}{In this section, we investigate the ``robustness" of the argument leading to Thm. \ref{thm: our GTT}. To begin, we will indeed revisit the derivation of Thm. \ref{thm: our GTT} for real and quaternionic Hilbert spaces. We consider two other structural changes: modifications of the tensor product used to describe composite systems on the one hand, and alternative state-update rules on the other. In doing so, we effectively discuss whether the extended Gleason theorem also applies to so-called \textit{foil theories} of quantum mechanics  \cite{chiribella_quantum_2016}.}

\subsection{\modification{Alternative number fields}} \label{sec:real_qt}
\textcolor{modification}{
Replacing the set $\mathbb{C}$ of complex numbers by the set of real numbers $\mathbb{R}$ in Postulate (M) and the word ``complex'' with ``real'' in Postulates (H) and (C) results in a foil theory known as \textit{real quantum theory} (RQT). This theory differs structurally from standard quantum theory: it is not locally tomographic \cite{wootters1990local}, does not respect monogamy of entanglement \cite{wootters2012entanglement} and, in certain network scenarios \cite{renou2021quantum}, cannot reproduce observed quantum correlations}.\footnote{\textcolor{modification}{These properties have led to the claim that quantum theory necessarily requires complex numbers \cite{hoffreumon2025quantum, hita2025quantum, ying2025whether}. Some authors argue that this conclusion is not warranted by ruling out a single toy theory such as RQT. Its shortcomings could be resolved by suitably redefining Postulate (C) used to describe composite systems in RQT \cite{hita2025quantum}.}}

\textcolor{modification}{These differences nonwithstanding, the derivation leading to the extended Gleason theorem of Sec.\ \ref{MarginalFramFunctions} remains valid. The lack of local tomography in RQT does not cause a problem since composite systems remain operationally equivalent to single systems of suitable dimension, even if some measurements cannot be interpreted as being performed locally. As before, composition compatibility (CC) fixes the form of the map $\Phi$ introduced in \ Eq.\ \eqref{CC in math form}, and the proofs of Lemma \ref{lem: d>2} and Thm.\ \ref{thm: our GTT} carry over.}

\textcolor{modification}{As in standard quantum theory, the partial trace is the unique mapping $\phi$ (cf.\ Eq.\ \ref{eq: gamma fcn (arbitrary g)}) that provides the induced preparation of a subsystem. Thus, the concept of marginal frame functions (see Def.\ \ref{def: CC frame function}) persists, allowing us to extend Gleason's theorem not only to standard qubits but also to \textit{rebits}, their equivalent in real Hilbert-space quantum theories.}


\modification{\textit{Quaternionic quantum theory} (QQT) is obtained by associating quantum systems with Hilbert spaces defined over the field $\mathbb{H}$ of quaternions \cite{finkelstein1962foundations}. When using a suitably defined \textit{real trace} rather than the standard trace \cite{moretti2018correct}, Gleason's theorem is found to hold for Hilbert spaces over $\mathbb{C}$, $\mathbb{R}$, and $\mathbb{H}$ of dimensions $d \geq 3$. This version reduces to the standard formulation of Thm.\ \ref{thm: GT} for real and complex Hilbert spaces.}

\modification{The composition of quaternionic systems is not straightforward since arbitrary quaternions do not commute. As a consequence, the standard tensor product \cite{finkelstein1962foundations, araki1980characterization} is not linear over quaternions. It has been argued that this property rules out \textit{local} operations in the quaternionic setting and, more generally,  that it prevents one from speaking of truly independent systems in QQT \cite{finkelstein1962foundations, graydon2013quaternionic}.}

\modification{To resolve this problem, modifications of the tensor product defining composite quaternionic quantum systems have been proposed \cite{ghiloni2017spectral, razon1992uniqueness}. They do not, however, seem to admit an analogue of Postulate (C) which underpins our extension of Gleason's theorem. Thus, the difficulty to consistently define composite systems appears to obstruct the extension of our result to QQT.}%
\subsection{\modification{Alternative sets of observables}} \label{sec:alternative_tensor_products}
\textcolor{modification}{Quantum theory uses a specific tensor product that determines both the set of \textit{states} in a composite quantum system and the set of \textit{measurements} that can be performed on them. Postulates (C) and (M) encode this choice which, ultimateley, reflects experimental evidence. Foil theories based on other tensor products \cite{chiribella_quantum_2016} have been considered to understand the particular features of the one apparently realized by Nature.} 

\textcolor{modification}{In the context of Gleason's theorem, quantum states are associated with consistent probability assignments to measurement outcomes for all observables. Alternative tensor products for composite systems come with modified sets of observables which may have consequences for the derivation of the original Gleason theorem and its extension to qubits. We will consider two different sets of measurements exclusively consisting of \textit{products} of projectors, a choice motivated by the fact that Lemma \ref{lem: d>2} and Thm.\ \ref{thm: our GTT} nothing but refer to operators of the form $P_x^A \otimes I^B$, etc.} 


\textcolor{modification}{Consider a foil theory in which experimenters are able to carry out only \textit{local} projective measurements. For a bi-partite system, the set of observables is then given by}
\begin{equation} 
    \textcolor{modification}{
    \mathsf{M}^{AB}_{\text{local}} =  \{ \{ P_i^A \otimes P_j^B \}_{i,j} :  \{ P_i^A  \}_{i} \in \mathsf{M}^A , \,  \{  P_j^B \}_{j} \in \mathsf{M}^B \}\,,  \label{localmmts}}
\end{equation}
\textcolor{modification}{hence not containing projectors on pure entangled states such as Bell states, for example. Due to the absence of projectors that cannot be written as a product with respect to the given bi-partition, this set is strictly smaller than that of standard quantum theory, $\mathsf{M}^{AB}_{\text{local}} \subset \mathsf{M}^{AB}$. The resulting \textit{minimal} tensor product of observables has been used, for example, to investigate quantum correlations in a scenario that involves only  ``locally quantum" experimenters \cite{barnum_2010_local}.}

\textcolor{modification}{Clearly, Gleason's theorem continues to apply (for $d>2$) to the factors $\mathcal{H}^A$ and $\mathcal{H}^B$ of the space $\mathcal{H}^A \otimes \mathcal{H}^B$ if considered on their own. Nevertheless, Proposition\ 5 in \cite{wallach2002unentangled} shows that the drastically reduced set of observables $\mathsf{M}^{AB}_{\text{local}}$ does not imply Gleason's theorem for the \textit{composite} system as a whole. There exist probability assignments to local measurements that correspond to non-quantum frame functions respecting $\mathsf{M}^{AB}_{\text{local}}$. 
} 

\textcolor{modification}{Next, we consider the largest set of observables containing only products of projectors,} 
\begin{equation}
    \textcolor{modification}{
    \mathsf{M}^{AB}_{\text{unent}} = \{ \{P_k\}_k \in \mathsf{M}^{AB} : \, P_k = P_k^A \otimes P_k^B \} \,.\label{unentangledmmts}}
\end{equation} 
\textcolor{modification}{It includes ``unentangled" observables that cannot be implemented locally \cite{bennett_1999_quantum} since they do not result from the direct tensor product of a pair of bases for the spaces $\mathcal{H}^A$ and $\mathcal{H}^B$, respectively. For two qubits, the set}
\begin{equation}
\textcolor{modification}{
    \{ \ket{0}\otimes\ket{0},\, \ket{0}\otimes\ket{1}, \,\ket{1}\otimes\ket{+},\, \ket{1}\otimes\ket{-} \}\,, \qquad \ket{\pm}=(\ket{0}\pm\ket{1})/\sqrt{2}\,, 
}
\end{equation}
\textcolor{modification}{represents a simple example of such a \textit{twisted} or \textit{indirect} orthonormal basis of states for the space $\mathbb{C}^2 \otimes \mathbb{C}^2$, constructed from more than two orthonormal bases. The set of unentangled PVMs clearly sits between the minimal and the standard quantum mechanical tensor product of observables,}
\begin{equation}
\textcolor{modification}{
    \mathsf{M}^{AB}_{\text{local}} \subset \mathsf{M}^{AB}_{\text{unent}} \subset \mathsf{M}^{AB}\,.}
\end{equation}
\textcolor{modification}{Unentangled product bases play a role in the classification of mutually unbiased bases \cite{mcnulty2012all, wiesniak2011entanglement}, for example.} 

\textcolor{modification}{Is it sufficient to define frame functions on the set of unentangled PVMs on composite systems to derive Born's rule and standard density matrices? Wallach has shown (see Thm.~2 of \cite{wallach2002unentangled}) that Gleason's theorem remains valid when considering this  considerably smaller set of observables, as long as each subsystem has a dimension of at least $d=3$. It turns out that it is not possible to apply our extension to Wallach's generalisation of Thm.\ \ref{thm: GT} for composite systems that contain at least one qubit Hilbert space. According to Thm.\ 3 of  \cite{wallach2002unentangled}, the product frame function $F( P_x^A \otimes P_y^B) \equiv f(P_x^A) g(P_y^B)$ respects $\mathsf{M}^{AB}_{\text{unent}}$ whenever $f \in \mathcal{F}_A$ and $g \in \mathcal{F}_B$ are frame functions associated with the subsystems $A$ and $B$, respectively. If $d_A =2$, one can use a non-quantum qubit frame function $f$  to define a frame function $F\in \mathcal{F}_{AB}$ \textit{not} associated with a proper quantum mechanical density matrix for the composite system. Trivially, the \textit{marginal} qubit frame function $f \in \mathcal{F}_2$ stemming from this non-quantum frame function $F\in \mathcal{F}_{AB}$ is also non-quantum, in view of Def.\ \ref{def: CC frame function}.}

\subsection{\modification{Alternative state-update rules} \label{sec:alternative_update_rules}}

\modification{Derivations of Gleason's theorem consider the probabilities of measurement outcomes as primitives. No assumptions are being made about the procedure(s) to access them experimentally. To establish this link in quantum theory, a more comprehensive Postulate (M) specifies the post-measurement states of projective measurements, known as L\"uders rule,}
\begin{equation}
    \modification{\rho \, \stackrel{x}{\longmapsto} \, P_x \, \rho \, P_x \, .}
\end{equation}
\modification{Having obtained the measurement outcome $x$ represented by $P_x$, the post-measurement state ``collapses" into the state associated with the outcome. Consequently, the outcome probabilities can only be obtained experimentally as limits of frequencies of repeated measurements on an ensemble of identically prepared systems.}

\modification{However, with the L\"uders rule being irrelevant for the derivation of Gleason's theorem (and its extension), the theorem will also hold in any foil theory that differs from quantum theory by nothing but an alternative state-update rule. Hypothetical theories of that type have been constructed in \cite{fiorentino_beyond_2025}, based on an operationally motivated notion of an \textit{update rule}. Among other requirements, acceptable update rules must assign post-measurement states consistently to both single and composite systems, and they must not allow superluminal signalling. The L\"uders projection satisfies all the requirements but, importantly, other unconventional update rules exist \cite{fiorentino_beyond_2025}. In \textit{passive quantum theory} \cite{fiorentino_quantum_2023}, for example, measurements do not alter the state of a system although the outcomes still occur probabilistically according to Born's rule. Importantly, the assumptions necessary to derive Gleason's theorem (and its extension) continue to hold in generalized state-update theories: they possess Hilbert-space representations with measurements given by PVMs, and system composition is defined by the Postulate (C).}

\section{\modification{Conclusions}} \label{conclusions}
\subsection{\modification{Summary}} \label{summary}

We have extended Gleason's theorem to the Hilbert space $\mathcal{H}_2$ by invoking only projective measurements. Instead of enlarging the types of measurements, we exploit the properties of measurements carried out on subsystems of a composite system.  

Local measurements on composite systems are part and parcel of standard quantum theory---the theory is actually incomplete without an axiom telling us how to describe multi-partite systems (cf.\ Postulate (C)). Considering bi-partite systems, we obtain consistency conditions for frame functions across different dimensions. In physical terms, the conditions express the requirement that the statistics of measurement outcomes cannot depend on whether a quantum system is considered on its own or as part of a larger composite system. 

Frame functions are said to be \textit{marginal} if they satisfy the consistency condition. For dimensions $d\geq 3$, all frame functions associated with projective measurements are found to be marginal by default. In dimension $d=2$, however, the constraint reduces the larger set of all frame functions (which contains \textit{non-quantum} probability assignments) to those (and only those) that admit a representation in terms of a density matrix. Thus, we recover the state space and the Born rule for probabilities for all finite-dimensional quantum systems including qubits, i.e.\ with dimensions $d \geq 2$. 

\modification{This argument extends beyond standard quantum theory. In real quantum theory, where composite systems are defined by the standard tensor product, the extension of Thm.\ \ref{thm: our GTT} still holds, even though the theory differs from standard quantum theory in composite scenarios.  Being independent of sequential measurements, the extended theorem remains valid upon modifications of the quantum  mechanical state-update. In contrast, altering the structure of the tensor product does affect our argument, even in scenarios where Gleason-type results can still be derived---unentangled PVMs suffice to recover density matrices for most composite systems but not to support our extension.
}

\subsection{\modification{Discussion}} \label{sum&disc}

Our extension of Gleason's theorem \modification{is of interest for at least four reasons. The derivation (i) is based exclusively on projection-valued measurements, (ii) relies on an additional assumption that is part of a comprehensive axiomatisation of quantum theory anyway, (iii) underlines the important role of composite systems for the structure of quantum theory, and (iv) opens up new perspectives regarding the reconstruction programme of quantum theory.} 

First, one may argue that the justification for PVMs \modification{in deriving Gleason's theorem} is \textit{logically} stronger \textcolor{modification}{than that for POVMs. Complete sets of mutually orthogonal projectors generate a Boolean lattice of events, so the additivity of probability over mutually exclusive outcomes follows directly from the logical structure. Mutual exclusiveness also reflects the physical interpretation of PVMs: prior to measurement, each projector corresponds to a property that the system either possesses or does not, linking outcomes to definite yes-no propositions. 
POVM elements also represent disjoint outcomes but do not form a Boolean lattice, i.e.\ they do not represent mutually exclusive properties. In a sense, POVMs are justified primarily on operational grounds: they capture all possible quantum measurements that can be realised physically when using ancillary systems and projective measurements acting on them. In this sense, PVMs can be viewed as a more primitive assumption.}

Second, an important feature of Thm.\ \ref{thm: our GTT} is the fact that its derivation \textcolor{modification}{directly follows from} another standard postulate of quantum mechanics.\footnote{An earlier attempt to extend Gleason's theorem to a two-dimensional Hilbert space without invoking non-projective measurements was proposed in \cite{de2016gleason}. However, the approach turned out to be flawed \cite{hall2016comment}.} The essential ingredient is the relation between the Hilbert space of a composite system and those of its constituents, i.e.\ Postulate (C). 

Composite systems have been studied before in the context of Gleason's theorem. 
\modification{As mentioned in Sec.\ \ref{sec:alternative_tensor_products},} an \textit{unentangled} variant of the theorem for multi-partite systems with system dimensions greater than two was proven in \cite{wallach2002unentangled}. More recently, Gleason's theorem has been shown to hold for the composition of two or more systems to which the theorem is known to apply separately but not jointly \cite{frembs2023gleason}. However, these approaches did not examine implications for the frame functions of a single qubit. 

A simplified derivation of Born's rule was given in \cite{logiurato2012born}, assuming from the outset that quantum states correspond to rays in Hilbert space. The proof for dimension $d \geq 3$ can be extended to $d=2$ by considering a qubit as part of a composite system. However, this extension only applies to composite systems that include a constituent of dimension $d \geq 3$. In contrast, our approach \textit{derives} density matrices for $d=2$ and applies to \textit{any} composite system with a qubit subsystem.

Third, let us \modification{add} a general remark concerning the importance of composite systems for the overall structure of quantum theory.\footnote{See also Sec.\ II of \cite{vollbrecht2000two}.} The notion of complete positivity \cite{davies1976quantum} only makes sense in multi-partite systems.\textcolor{modification}{\footnote{\textcolor{modification}{Although the Kraus representation characterises complete positivity at the single-system level, the justification for excluding linear maps that do not allow a Kraus representation (e.g.\ the partial transpose) is inherently multi-partite. A map $\omega$ acting on system $A$ must remain positive when extended as $\omega \otimes \mathcal{I}^B$ to any ancillary system $B$, where $\mathcal{I}^B$ is the identity channel on $\mathcal{H}^B$.}}} Hidden variable models exist for a single qubit but local ones can be ruled out when considering composite systems \cite{bell1964einstein, bell1966problem}. It is also known \cite{auffeves2020deriving} that the only acceptable probability assignments to orthogonal projections in any two-dimensional subspace of $\mathbb{C}^3$ are those stemming from valid quantum mechanical density matrices. However, this property does not prevent one from assigning non-quantum probabilities to projections in the space $\mathbb{C}^2$ considered on its own. 

Studies of foil theories with composition rules differing from the tensor product reveal the extent to which the properties of quantum systems depend on the standard choice \cite{janotta2014generalized}. Similarly, our proof of Thm.\ \ref{thm: our GTT}, which extends Gleason's theorem to $d=2$, relies entirely on the description of composite quantum systems in terms of the tensor product. Thus, our result illustrates in yet another way how important this composition rule is for the structure of quantum theory as a whole.

\modification{Finally, extending Gleason's theorem plays a role in the reconstruction programme of quantum theory \cite{hardy_quantum_2001, grinbaum2007reconstruction, chiribella_informational_2011,masanes_derivation_2011, jaeger2019information}. Its goal is to make characteristic quantum features such as the L\"uders rule plausible, if not to derive the entire theory through natural physical assumptions.}

\modification{More specifically, the extended theorem allows one to derive the structure of the state space and the Born rule for systems of \textit{arbitrary} dimensions, simply by combining the postulates for system composition (C) and (projective) measurements (M). Next, the unitarity of quantum dynamics is recovered from the Hilbert-space setting via Wigner-Uhlhorn theorems \cite{wigner1931gruppentheorie, uhlhorn_1963_representation} or alternative approaches \cite{wilson_origin_2023}. Once unitaries have been established 
the L\"uders rule can in turn be derived following arguments given in Ref.\ \cite{fiorentino_beyond_2025}, the framework of theories with generalized state-update rules described in Sec. \ref{sec:alternative_update_rules}. 
}

\section*{Acknowledgements}
We would like to thank M.\ J.\ W.\ Hall, L.\ Loveridge and T.\ Vetterlein for their constructive comments on a draft of this paper. \modification{We also thank an anonymous referee for suggesting to explore generalizations of Thm.\ \ref{thm: our GTT}.} This work was supported by the Leverhulme Trust, grant number RPG-2024-201. 

\printbibliography

\end{document}